\documentclass{article}
\usepackage{amssymb,amsthm}

\newtheorem{Prop}{Proposition}[section]
\newtheorem*{Prop321}{Proposition III.1.1}
\newtheorem*{Prop322}{Proposition III.1.2}
\newtheorem*{Prop33}{Proposition III.2}

\newcommand{\be}{\begin{equation}}
\newcommand{\ee}{\end{equation}}
\newcommand{\lb}{\label}
\newcommand{\ol}{\overline}

\newcommand{\ba}{{\bf a}}

\newcommand{\bp}{{\bf p}}

\newcommand{\bu}{{\bf u}}
\newcommand{\bv}{{\bf v}}
\newcommand{\bw}{{\bf w}}
\newcommand{\bx}{{\bf x}}
\newcommand{\bz}{{\bf z}}

\newcommand{\bA}{{\bf A}}
\newcommand{\bB}{{\bf B}}
\newcommand{\bE}{{\bf E}}
\newcommand{\bR}{{\bf R}}
\newcommand{\bU}{{\bf U}}
\newcommand{\bW}{{\bf W}}
\newcommand{\bJ}{{\bf J}}
\newcommand{\bS}{{\bf S}}

\newcommand{\wt}{\widetilde}

\newcommand{\bomega}{{\mbox{\boldmath $\omega$}}}

\newcommand{\bOmega}{{\mbox{\boldmath $\Omega$}}}

\newcommand{\grad}{{\mbox{\boldmath $\nabla$}}}
\newcommand{\bdot}{{\mbox{\boldmath $\cdot$}}}
\newcommand{\bcirc}{{\mbox{{\tiny \boldmath $\circ$}}}}

\newcommand{\btimes}{{\mbox{\boldmath $\times$}}}
\newcommand{\bzed}{{\mbox{\boldmath $0$}}}

\title{Stochastic Line-Motion and Stochastic Conservation Laws
for Non-Ideal Hydromagnetic Models.  I. Incompressible Fluids and
Isotropic Transport Coefficients}

\author{Gregory L. Eyink\\
{\it Department of Applied Mathematics \& Statistics}\\
{\it The Johns Hopkins University}\\
{\it Baltimore, Maryland, USA}
}
\date{}

\begin{document}

\maketitle

\begin{abstract}
We prove that smooth solutions of non-ideal (viscous and resistive)
incompressible
magnetohydrodynamic equations  satisfy a stochastic law of flux conservation.
This property involves an ensemble of surfaces obtained from a given, fixed
surface by advecting it backward in time under the plasma velocity perturbed
with a random white-noise. It is shown that the magnetic flux through the
fixed surface is equal to the average of the magnetic fluxes through
the ensemble of surfaces at earlier times. This result is an analogue
of the well-known Alfv\'{e}n theorem of ideal MHD and is valid for any value
of the magnetic Prandtl number. A second stochastic conservation law is shown
to hold
at unit Prandtl number, a random version of the generalized Kelvin theorem
derived
by Bekenstein-Oron for ideal MHD. These stochastic conservation laws are not
only
shown to be consequences of the non-ideal MHD equations, but are proved in fact
to be equivalent to those equations.  We derive similar results for two more
refined
hydromagnetic models, Hall magnetohydrodynamics and the two-fluid plasma model,
still assuming incompressible velocities and isotropic transport coefficients.
Finally,
we use these results to discuss briefly the infinite-Reynolds-number limit of
hydromagnetic turbulence and to support the conjecture that flux-conservation
remains
stochastic in that limit.
 \end{abstract}

 \newpage

\section{Introduction}\lb{Intro}

If a plasma is sufficiently collisional, then it can be well-described by
fluid-mechanical
equations. There is a hierarchy of such hydromagnetic models, ranging from
standard
magnetohydrodynamics (MHD), to more refined models such as Hall MHD and the
two-fluid model, with separate equations for electron and ion fluids
\cite{Kulsrud05,Bellan06}.
In all of these fluid models the magnetic field lines (or magneto-vortex lines)
at the limit
of infinite conductivity  are ``frozen-in'' to the plasma, as first observed by
Alfv\'{e}n
\cite{Alfven42}. The properties of magnetic fields for MHD-type models are
closely
analogous to the properties of vorticity fields for the Navier-Stokes equation
of neutral
fluids \cite{Newcomb58,Stern66}.  Thus, the Helmholtz-Kelvin theorem on
conservation of circulations \cite{Helmholtz1858,Thomson1869} has an MHD
analogue in the Alfv\'{e}n theorem on flux conservation \cite{Alfven42}, and
the
Cauchy formula for vorticity \cite{Cauchy1815} has an analogue in the Lundquist
formula for magnetic field \cite{Lundquist51}. These Lagrangian properties of
magnetic fields are central to many physical processes in plasmas, such as
magnetic
dynamo and magnetic reconnection. It has been claimed, with some justification,
that  ``The most important property of an ideal plasma is flux freezing''
\cite{Kulsrud05}
(section 3.2).

The fundamental Lagrangian laws of conservation and line-motion hold exactly
only for smooth solutions of ideal fluid equations, with zero viscosities
and resistivities.  Recently, however, it has been shown by Constantin and Iyer
\cite{ConstantinIyer08,Iyer06a,Iyer06b} that the analogous laws of vorticity
under
ideal Euler evolution remain for the viscous Navier-Stokes solution as {\it
stochastic}
laws.  To formulate these results, the equations for Lagrangian fluid particles
advected
by the Navier-Stokes velocity field must be perturbed by a Gaussian
white-noise,
with amplitude depending upon the viscosity. A random ensemble of fluid motions
results,
depending upon the realization of the white-noise process. To calculate the
fluid circulation
on a given closed loop in the fluid, the loop is evolved backward in time to
obtain
a random ensemble of loops. The circulation on the given loop is the {\it
average}
over the circulations of the ensemble of loops at the earlier time. Not only do
Navier-Stokes solutions enjoy this remarkable ``stochastic Kelvin theorem,"
or conservation of circulations in the mean,  but Constantin and Iyer
\cite{ConstantinIyer08,Iyer06a,Iyer06b} have shown that this property
also uniquely characterizes the velocity fields which satisfy the Navier-Stokes
equation. Furthermore, vortex-lines at any chosen initial time are
``frozen-in'' to
the stochastic fluid flows and thus become themselves stochastic. The resultant
(deterministic) vorticity at any point at a later time is the average over the
random
ensemble of vorticity vectors that are advected to that point, stretched and
tilted,
by the stochastic flows. These results provide an intuitive way to understand
vortex
dynamics, both stretching and reconnection, in viscous Navier-Stokes fluids.

In this paper, we demonstrate similar stochastic conservation laws and
``frozen-in'' properties
for non-ideal (resistive and viscous) plasma fluid models. The proofs of
Constantin and
Iyer \cite{ConstantinIyer08,Iyer06a,Iyer06b} exploit the Weber formulation
\cite{Weber1868}
of the incompressible Euler equations, and we shall employ here the similar
Weber formulations
of incompressible hydromagnetic models developed in the work of Ruban and
Kuznetsov
\cite{Ruban99,KuznetsovRuban00}. Such Weber formulas are implied by the
Hamiltonian
structure of the ideal fluid models and thus have considerable generality. In
the case of
standard MHD, the Weber formula of \cite{KuznetsovRuban00} is mathematically
equivalent
to the generalized Kelvin theorem of Bekenstein and Oron
\cite{BekensteinOron00}, which
provides a second Lagrangian conservation law in addition to the Alfv\'{e}n
theorem
on flux conservation. We shall prove that non-ideal MHD at unit magnetic
Prandtl number
enjoys analogues of both of these results as stochastic conservation laws. For
general Prandtl
number, we shall show that MHD retains at least the stochastic Alfv\'{e}n
theorem, or
flux-conservation in the average sense. This result may be formulated
equivalently as a
stochastic Lundquist formula, according to which the magnetic field vectors are
``frozen-in''
to the stochastic flows and then ensemble-averaged to yield the resultant
magnetic field.
Similar results shall be established also for two more refined hydromagnetic
models, given
by the Hall MHD equations and the two-fluid model of electrons and ions.

There are several further directions in which these methods and ideas may be
developed.
It has been shown that the Constantin-Iyer formulation of the incompressible
Navier-Stokes equation in fact corresponds to a variational principle, a
stochastic version
of the Hamilton-Maupertuis principle of least-action \cite{Eyink09}. In that
work the stochastic
Kelvin theorem was shown to arise from a symmetry of the stochastic action
under the
infinite-dimensional particle-relabelling group.  The ideal fluid models of
hydromagnetics
are also Hamiltonian in form, as discussed in
\cite{Ruban99,KuznetsovRuban00,BekensteinOron00}
and references therein.  The results of the present paper can be derived from
stochastic
action principles and, in particular, the stochastic version of the
Bekenstein-Oron generalized
Kelvin theorem for MHD can be shown to arise from invariance of the stochastic
MHD action
under particle-relabelling. These results shall be given elsewhere. We shall
also extend the
main results of this paper in following work  \cite{EyinkNeto09} to
hydromagnetic models of compressible
fluids with anisotropic transport coefficients, i.e. with differing values of
viscosity and resistivity in directions
longitudinal and transverse to the magnetic field. Such refinements have
importance in applications,
but we confine ourselves here for simplicity to incompressible fluids and
isotropic coefficients.
We shall also make only a few brief remarks, in the conclusion section, about
turbulent
hydromagnetics at very high (kinetic and magnetic) Reynolds numbers
\cite{Eyink06,Eyink07,Eyink08,Constantin08}.

\section{Standard Magnetohydrodynamics}\lb{MHD}

We consider in this section the standard incompressible MHD equations
\cite{Kulsrud05,Bellan06}
for the velocity field $\bu$ and the magnetic field $\bB$ in three
space-dimensions,
written (in cgs units) as:
\be \partial_{t} \bu + (\bu\bdot\grad)\bu = -\grad p + \frac{1}{\rho
c}\bJ\btimes\bB
+ \nu\bigtriangleup\bu
\lb{u-eq} \ee
\be \partial_{t} \bB = \grad\times(\bu\btimes\bB) + \lambda\bigtriangleup\bB
\lb{B-eq} \ee
\be \grad\cdot\bB=\grad\cdot\bu=0
\lb{div-free} \ee
Here the mass density $\rho,$ the kinematic viscosity $\nu=\mu/\rho$ and the
magnetic diffusivity
$\lambda=\eta c/4\pi$ are all
assumed to be space-time constants. The electric current $\bJ$ is given by the
nonrelativistic
approximation to Ampere's law as $\bJ= \frac{c}{4\pi}(\grad\btimes\bB)$. Since
the speed of light
$c$ cancels in eq.(\ref{u-eq}), it is simplifying to set $c=1$ and it is also
convenient to assume that
density $\rho=1.$ We do both throughout this section.

As did Constantin-Iyer \cite{ConstantinIyer08,Iyer06a,Iyer06b}, we take the
flow domain $\Omega$
to be either the 3-torus ${\mathbb T}^3$ or else 3-dimensional Euclidean space
${\mathbb R}^3.$
In the latter case, we require that $\bu,\bB$ decay sufficiently at infinity.
The existence of weak
solutions to the above MHD system (\ref{u-eq})-(\ref{div-free}) and the
existence, uniqueness,
and regularity of local-in-time strong solutions are established, for example,
in
\cite{DuvautLions72,SermangeTemam83}. We shall consider here only sufficiently
smooth initial data $\bu_0,\bB_0\in C^{k,\alpha}(\Omega)$ with $k\geq 3,$ where
$C^{k,\alpha}(\Omega)$ for $k\geq 0$ and $\alpha\in (0,1)$ is the Banach space
of functions
that are $k$-times differentiable with $k$th partial-derivatives H\"{o}lder
continuous of
exponent $\alpha.$ We shall use the fact that there exists for such initial
data a unique
solution of (\ref{u-eq})-(\ref{div-free}) with $\bu,\bB\in
C([t_0,t_f],C^{k,\alpha}(\Omega),$
$k\geq 2,$ for some $T=t_f-t_0>0.$ It should be possible as in
\cite{Iyer06a,Iyer06b}
to give a self-contained local-existence result, based upon the fixed-point
characterization
of MHD solutions in the theorems below. However, here we find it simpler to
give a more
direct argument that existing smooth solutions possess the stated stochastic
conservation
laws.


\subsection{Unit Magnetic Prandtl Number}

Our first result shall be for the case of unit magnetic Prandtl number, when
$\nu=\lambda.$
In this case, we can prove that there are {\it two} stochastic Lagrangian
conservation laws
for solutions of (\ref{u-eq})-(\ref{div-free}), one corresponding to the
Alfv\'{e}n theorem
\cite{Alfven42} and another corresponding to a generalized Kelvin theorem
\cite{KuznetsovRuban00,BekensteinOron00}. As in
\cite{ConstantinIyer08,Iyer06a,Iyer06b},
we shall show that the stochastic conservation laws furthermore uniquely
characterize
the MHD solutions. A precise statement of our result is as follows:

\begin{Prop}\lb{Prop1}
Divergence-free fields  $\bu,\bB\in C([t_0,t_f],C^{k,\alpha}(\Omega))$  satisfy
the non-ideal,
incompressible MHD equations (\ref{u-eq})-(\ref{B-eq}) with initial data
$\bu_0,\bB_0\in
C^{k,\alpha}(\Omega)$ for $k\geq 3$ and $\nu=\lambda,$  iff for all closed,
rectifiable loops $C$ and for all $t\in [t_0,t_f]$ (with $\bA={{\rm
curl}}^{-1}\bB)$
\be \oint_C \bA(\bx,t)\bdot d\bx=
     {\mathbb E} \left[\oint_{\wt{\ba}(C,t)} \bA_0(\ba)\bdot d\ba \right],
\lb{Alfven-A-thm} \ee
\be \oint_C \bu(\bx,t)\bdot d\bx=
     {\mathbb E} \left[\oint_{\wt{\ba}(C,t)} [\bu_0(\ba)+\bB_0(\ba)\btimes
\wt{\bR}_*(\ba,t)]\bdot d\ba \right].
 \lb{gen-Kelvin} \ee
Here $\wt{\ba}(\bx,t)$ are ``back-to-label maps'' for stochastic forward flows
$\wt{\bx}(\ba,t)$ solving
\be d\wt{\bx}(\ba,t)=\bu(\wt{\bx}(\ba,t),t)dt+\sqrt{2\nu}\,d\bW(t),\,\,t>t_0,
\,\,\,\,\,\,\,\,\,\,\,\,\,\,\, \wt{\bx}(\ba,t_0)=\ba, \lb{forward} \ee
$\wt{\bR}_*(\ba,t)$ is the {\it Lagrangian-history charge density} (charge per
unit area) satisfying
\be \partial_t \wt{\bR}_*(\ba,t) =
-\bJ(\wt{\bx}(\ba,t),t)(\grad_a\wt{\bx}(\ba,t))^{-1}, \,\,t>t_0,
\,\,\,\,\,\,\,\,\,\,\,\,\,\,\, \wt{\bR}_*(\ba,t_0)=\bzed, \lb{Rstar-eq} \ee
and ${\mathbb E}$ in (\ref{Alfven-A-thm}),(\ref{gen-Kelvin}) denotes average
over realizations
of the Brownian motion $\bW(t)$ in the SDE (\ref{forward}).
\end{Prop}

\noindent Before we present the proof of the above proposition, let us make a
few explanatory remarks.
The first result (\ref{Alfven-A-thm}) may be re-expressed in terms of a flux
integral through a smooth
bounding surface $S$ of the closed loop $C$, with $C=\partial S.$ It takes the
form of a
{\it stochastic Alfv\'{e}n theorem}, expressing conservation on average of
magnetic flux:
\be \int_S \bB(\bx,t)\bdot d\bS(\bx)=
     {\mathbb E} \left[\int_{\wt{\ba}(S,t)} \bB_0(\ba)\bdot d\bS(\ba) \right],
\,\,\,\,\,\, \lb{Alfven-B-thm} \ee
This result is, in turn, equivalent to a {\it stochastic Lundquist formula} for
the local magnetic field,
\be \bB(\bx,t)=
{\mathbb E} \left[\left.\bB_0(\ba)\bdot \grad_a\wt{\bx}(\ba,t)
           \right|_{\wt{\ba}(\bx,t)} \right]. \lb{Lundquist} \ee
The second law (\ref{gen-Kelvin}) is a {\it stochastic Kelvin theorem},
expressing conservation
on average of generalized circulation, analogous to the deterministic result of
\cite{BekensteinOron00}.
It is equivalent to a {\it stochastic Weber formula}, corresponding to the
deterministic formula of
\cite{Ruban99,KuznetsovRuban00},
\be \bu(\bx,t) =
{\mathbb E }{\mathbb P}\left[ \grad_x\wt{\ba}(\bx,t)\left.
  \left(\bu_0(\ba)+\bB_0(\ba)\btimes
\wt{\bR}_*(\ba,t)\right)\right|_{\wt{\ba}(\bx,t)}\right].
  \lb{gen-Weber-L} \ee
where ${\mathbb P}$ denotes the Leray-Hodge projection onto divergence-free
vector fields
\cite{ChorinMarsden00}.
To explain the physical meaning of this formula, it is useful to transform the
final term
to Eulerian form:
\be \bu(\bx,t) =
{\mathbb E }{\mathbb P}\left[ \grad_x\wt{\ba}(\bx,t)\bu_0(\wt{\ba}(\bx,t))
           +\wt{\bB}(\bx,t)\btimes \wt{\bR}(\bx,t)\right], \lb{gen-Weber-E} \ee
where $\wt{\bB}(\bx,t)\equiv \left.\bB_0(\ba)\bdot\grad_a\wt{\bx}(\ba,t)
\right|_{\wt{\ba}(\bx,t)}$ and $\wt{\bR}(\bx,t)\equiv
\left.\wt{\bR}_*(\ba,t)\bdot
\grad_a\wt{\bx}(\ba,t)\right|_{\wt{\ba}(\bx,t)}.$ Then, employing the usual
rules of calculus,
$\wt{\bR}(\bx,t)$ satisfies the Stratonovich SPDE which follows from
(\ref{Rstar-eq}):
\be d\wt{\bR}(\bx,t)=\grad\btimes\left(\bU(\bx,\bcirc
\,dt)\btimes\wt{\bR}(\bx,t)\right)
     -\bJ(\bx,t)dt, \,\,t>t_0,\,\,\,\,\,\,\,\,\,\,\,\,\,\,\,
\wt{\bR}(\bx,t_0)=\bzed, \lb{R-eq} \ee
with $\bU(\bx,t)=\int_{t_0}^t dt'\,\bu(\bx,t')+\sqrt{2\nu}\bW(t)$ the (Ito and
Stratonovich)
infinitesimal generator of the stochastic Lagrangian flow. Consider the flux
integral
of $\wt{\bR}_*$ through any smooth surface:
\be \wt{Q}(S,t)\equiv \int_S \wt{\bR}_*(\ba,t)\bdot d\bS(\ba)
                         =\int_{\wt{\bx}(S,t)} \wt{\bR}(\bx,t)\bdot d\bS(\bx).
\lb{R-flux} \ee
Then, (\ref{R-eq}) is equivalent to the following equation, valid for all
smooth surfaces $S$:
\be d\wt{Q}(S,t)=-dt\cdot \int _{\wt{\bx}(S,t)} \bJ(\bx,t)\bdot d\bS(\bx),
\,\,t>t_0,
\,\,\,\,\,\,\,\,\,\,\,\,\,\,\,\wt{Q}(S,t_0)=0. \lb{Q-eq} \ee
The above equation implies that $-\wt{Q}(S,t)$ equals the electric charge which
flowed across the material surface $\wt{\bx}(S,t)$ between times $t_0$ and $t.$
This
explains the name ``Lagrangian-history charge density'' for the field
$\bR_*(\ba,t)$
used in the above proposition.  For any infinitesimal vector surface element
$d\bS(\ba)$
starting at point $\ba,$ $-\bR_*(\ba,t)\bdot d\bS(\ba)$ equals the charge
crossing
the advected surface from $t_0$ to $t.$

\begin{proof}[Proof of Proposition \ref{Prop1}:] We first remark that, given
the velocity
field $\bu\in C([t_0,t_f],C^{k,\alpha}(\Omega))$ for $k\geq 1,$ there exists a
stochastic
flow $\wt{\bx}(\ba,t)$ of $C^{k,\alpha}$-diffeomorphisms solving the SDE
(\ref{forward}),
so that the inverse map $\wt{\ba}(\bx,t)$ exists and belongs to
$C([t_0,t_f],C^{k,\alpha}
(\Omega)),$ at least for $T$ sufficiently small. This follows by the arguments
in \cite{Iyer06a,Iyer06b} or the general methods in the monograph
\cite{Kunita90},
Chapter 4. If we assume that $\bu,\bB\in C([t_0,t_f],C^{k,\alpha}(\Omega))$ for
$k\geq 3,$
then $\grad_a\wt{\bx},\bJ\in C([t_0,t_f],C^{k,\alpha}(\Omega))$ for $k\geq 2,$
so that
$\wt{\bR}_*\in C^1([t_0,t_f],C^{k,\alpha}(\Omega))$ for $k\geq 2.$ This is
sufficient regularity
to justify all of our calculations below. It is particularly important that
$\wt{\bR}_*$ is bounded
variation in time and has no martingale part.

We begin by showing the ``if'' direction. Therefore, assuming that
divergence-free fields
$\bu,\bB \in C([t_0,t_f],C^{k,\alpha}(\Omega))$ with $k\geq 3$ solve the
fixed-point problem
(FPP) specified by the eqs.(\ref{Alfven-A-thm})--(\ref{Rstar-eq}) we must prove
that they also
satisfy the incompressible MHD equations (\ref{u-eq})--(\ref{B-eq}). The
argument
closely follows the proof of Theorem 2.2 in Section 4 of
\cite{ConstantinIyer08}.
We make the successive definitions
\be \wt{\bw}(\bx,t) =
  \left[\bu_0(\ba)+\bB_0(\ba)\btimes \wt{\bR}_*(\ba,t)\right]_{\wt{\ba}(\bx,t)}
\lb{w-def} \ee
\be \wt{\bv}(\bx,t) = \grad_x\wt{\ba}(\bx,t)\wt{\bw}(\bx,t) \lb{v-def} \ee
\be \wt{\bu}(\bx,t) = {\mathbb P}\wt{\bv}(\bx,t) =
\wt{\bv}(\bx,t)-\grad_x\wt{\varphi}(\bx,t)
      \lb{u-def} \ee
We see that $\wt{\bw},\wt{\bv},\wt{\bu}\in C([t_0,t_f],C^{k,\alpha}(\Omega))$
for $k\geq 2$
and that the stochastic Weber formula (\ref{gen-Weber-L}) is restated as $\bu=
{\mathbb E}(\wt{\bu}).$  We now develop a stochastic evolution equation for
each of
these three variables.

Note first that the ``back-to-labels map'' $\wt{\ba}$ satisfies
\be d\wt{\ba}(\bx,t)
+\left[(\bu\bdot\grad_x)\wt{\ba}-\nu\bigtriangleup\wt{\ba}\right]dt
              +\sqrt{2\nu}(d\bW(t)\bdot\grad_x)\wt{\ba} =0. \lb{a-eq} \ee
This is proved in \cite{ConstantinIyer08}, Proposition 4.2. It can also be
derived as a special case
of the``first It\^{o} formula'' for a backward flow; see \cite{Kunita90},
Theorem 4.4.5.
Then, as a consequence of the generalized Ito rule,
\be d \wt{\bw}(\bx,t) =\left[
-(\bu\bdot\grad_x)\wt{\bw}+\nu\bigtriangleup\wt{\bw}\right]dt
       + \left[\bB_0(\ba)\btimes
\partial_t\wt{\bR}_*(\ba,t)\right]_{\wt{\ba}(\bx,t)}dt
       -\sqrt{2\nu}(d\bW(t)\bdot\grad_x)\wt{\bw} \lb{w-eq} \ee
For example, see \cite{ConstantinIyer08}, Corollary 4.3. The term in
(\ref{w-eq}) involving
$\partial_t\wt{\bR}_*$ can be evaluated using
\be \grad_x\wt{\ba}(\bx,t)\left[\bB_0(\ba)\btimes
\partial_t\wt{\bR}_*(\ba,t)\right]_{\wt{\ba}(\bx,t)}
         = \bJ(\bx,t)\btimes \wt{\bB}(\bx,t), \lb{Rdot-eq} \ee
which follows from (\ref{Rstar-eq}), and from the definition
\be \wt{\bB}(\bx,t)=\left.\bB_0(\ba)\bdot
\grad_a\wt{\bx}(\ba,t)\right|_{\wt{\ba}(\bx,t)}.
 \lb{B-def} \ee
We next calculate the differential of $\wt{\bv}$ using the Ito product rule,
$$  d \wt{\bv}(\bx,t)  = \grad_x\wt{\ba}(\bx,t)d\wt{\bw}(\bx,t)
      + d(\grad_x\wt{\ba})\wt{\bw}(\bx,t)+
d\langle\grad_x\wt{\ba},\wt{\bw}\rangle, $$
which, together with (\ref{a-eq}),(\ref{w-eq}),(\ref{Rdot-eq}), gives
\be   d \wt{\bv}(\bx,t) =\left[ -(\bu\bdot\grad_x)\wt{\bv}
        -(\grad_x\bu)\wt{\bv}+\bJ(\bx,t)\btimes
\wt{\bB}(\bx,t)+\nu\bigtriangleup\wt{\bv}\right]dt
              -\sqrt{2\nu}(d\bW(t)\bdot\grad_x)\wt{\bv}. \lb{v-eq} \ee
The rest of the argument goes exactly as in Section 4 of
\cite{ConstantinIyer08}. As in
the proof of Theorem 2.2 of \cite{ConstantinIyer08}, the differential of
$\wt{\bu},$ as defined
in (\ref{u-def}), can be expressed as
\begin{eqnarray}
 d \wt{\bu}(\bx,t) & = & \left[ -(\bu\bdot\grad_x)\wt{\bu}
        -(\grad_x\bu)\wt{\bu}+\bJ(\bx,t)\btimes
\wt{\bB}(\bx,t)+\nu\bigtriangleup\wt{\bu}\right]dt
              -\sqrt{2\nu}(d\bW(t)\bdot\grad_x)\wt{\bu} \cr
       & & \,\,\,\,\,\,\,\,\,\,\,\,\,\,\, -\grad_x\left[d
\wt{\varphi}+\left((\bu\bdot\grad_x)\wt{\varphi}
-\nu\bigtriangleup\wt{\varphi}\right)dt
+\sqrt{2\nu}(d\bW(t)\bdot\grad_x)\wt{\varphi} \right]
\lb{utilde-eq} \end{eqnarray}
and, using again the generalized Ito rule and definition (\ref{B-def}),
\be  d \wt{\bB}(\bx,t) =\left[
-(\bu\bdot\grad_x)\wt{\bB}+(\wt{\bB}\bdot\grad_x)\bu
       +\nu\bigtriangleup\wt{\bB}\right]dt
-\sqrt{2\nu}(d\bW(t)\bdot\grad_x)\wt{\bB}
       \lb{Btilde-eq} \ee
just as in the proof of Proposition 2.7 of \cite{ConstantinIyer08}. Taking the
expectation
over the Brownian motion in eqs.(\ref{utilde-eq})-(\ref{Btilde-eq}) yields
eqs.(\ref{u-eq})-(\ref{B-eq}) with $\lambda=\nu$ and kinematic pressure
$p=\frac{1}{2}|\bu|^2+\dot{\varphi}+(\bu\bdot\grad_x)\varphi
-\nu\bigtriangleup\varphi.$

We finally show the  ``only if'' direction. Therefore, assuming that
divergence-free fields
$\bu,\bB \in C([t_0,t_f],C^{k,\alpha}(\Omega))$ with $k\geq 3$ solve the
incompressible
MHD equations (\ref{u-eq})--(\ref{B-eq}) we shall show that they also satisfy
the fixed-point
problem (FPP) specified by the eqs.(\ref{Alfven-A-thm})--(\ref{Rstar-eq}). Let
us define
\be \ol{\bu}(\bx,t) =
{\mathbb E }{\mathbb P}\left[ \grad_x\wt{\ba}(\bx,t)\left.
  \left(\bu_0(\ba)+\bB_0(\ba)\btimes \wt{\bR}_*(\ba,t)\right)
  \right|_{\wt{\ba}(\bx,t)}\right] \lb{ubar-def} \ee
and
\be \ol{\bB}(\bx,t)=
{\mathbb E} \left[\left.\bB_0(\ba)\bdot \grad_a\wt{\bx}(\ba,t)
           \right|_{\wt{\ba}(\bx,t)} \right], \lb{Bbar-def} \ee
where $\wt{\bx}(\ba,t)$ solves (\ref{forward}) and $\wt{\bR}_*(\ba,t)$ solves
(\ref{Rstar-eq}),
for the given $\bu,\bB.$ It then follows from our previous work that
$\ol{\bu},\ol{\bB}
\in C([t_0,t_f],C^{k,\alpha}(\Omega))$ with $k\geq 2,$ are divergence-free, and
solve
the linear equations
\be  \partial_t \ol{\bu} =  -(\bu\bdot\grad)\ol{\bu}
        -(\grad\bu)\ol{\bu}-\grad\ol{p}+\bJ\btimes \ol{\bB}
+\nu\bigtriangleup\ol{\bu},
        \lb{ubar-eq} \ee
\be \partial_t \ol{\bB} =  -(\bu\bdot\grad)\ol{\bB}
        + \ol{\bB}\bdot\grad\bu+\nu\bigtriangleup\ol{\bB}, \lb{Bbar-eq} \ee
with initial conditions $\ol{\bu}(t_0)=\bu_0,\ol{\bB}(t_0)=\bB_0.$  At least
one solution
is the pair $(\bu,\bB)$ itself, so that, if solutions of the initial-value
problem are unique,
it must be the case that $(\ol{\bu},\ol{\bB})=(\bu,\bB).$ It thus suffices to
prove that the
linear system (\ref{ubar-eq})-(\ref{Bbar-eq}) has unique solutions for
specified initial data.

This may be shown by a standard argument based on an energy estimate (e.g. see
\cite{Nunez97}). For the pair $\bz(\bx)=(\bu(\bx),\bB(\bx))$ define norms
$$ \|\bz\|_2 = \left(\int_\Omega d^3x\,\left[|\bu(\bx)|^2+|\bB(\bx)|^2\right]
\right)^{1/2},\,\,\,\,\,\,\,
\|\bz\|_\infty = \sup_{\bx\in\Omega}\left[|\bu(\bx)|+|\bB(\bx)|\right]. $$
An easy calculation then shows for any solution
$\ol{\bz}(\bx,t)=(\ol{\bu}(\bx,t),
\ol{\bB}(\bx,t))$ of (\ref{ubar-eq}),(\ref{Bbar-eq}) that
\begin{eqnarray*}
\frac{d}{dt} \|\ol{\bz}(t)\|_2^2 & = & 2\int_\Omega d^3 x\, \left[\partial_j
u_i(\bx,t)\left(
     \ol{B}_i(\bx,t) \ol{B}_j(\bx,t)-\ol{u}_i(\bx,t)
\ol{u}_j(\bx,t)\right)\right. \cr
     && \,\,\,\,\,\,\,\,\,\,\,\,\,\,\,\,\,\,\,\,\,\,\left. +\partial_j
B_i(\bx,t)\left(\ol{u}_i(\bx,t) \ol{B}_j(\bx,t)
     -\ol{u}_j(\bx,t)
\ol{B}_i(\bx,t)\right)\right]-2\nu\|\grad\ol{\bz}(t)\|_2^2 \cr
     & \leq & 2\|\grad\bz(t)\|_\infty
\|\ol{\bz}(t)\|_2^2-2\nu\|\grad\ol{\bz}(t)\|_2^2.
\end{eqnarray*}
For some $\epsilon>0,$ choose $\gamma>\sup_{t\in
[t_0,t_f]}\|\grad\bz(t)\|_\infty + \epsilon.$
Then it follows that
$$ \frac{d}{dt} \left[e^{-2\gamma(t-t_0)}\|\ol{\bz}(t)\|_2^2\right]
      \leq -2 e^{-2\gamma(t-t_0)}\left[
\epsilon\|\ol{\bz}(t)\|_2^2+\nu\|\grad\ol{\bz}(t)\|_2^2\right]. $$
Integration yields the energy inequality
\be e^{-2\gamma(t_f-t_0)}\|\ol{\bz}(t_f)\|_2^2+
       2\int_{t_0}^{t_f} dt\,  e^{-2\gamma(t-t_0)}\left[
\epsilon\|\ol{\bz}(t)\|_2^2+\nu\|\grad\ol{\bz}(t)\|_2^2\right]
       \leq   \|\bz_0\|_2^2, \lb{z-ineq} \ee
which implies uniqueness of solutions of (\ref{ubar-eq}),(\ref{Bbar-eq}) as a
direct consequence.
\end{proof}

\subsection{General Magnetic Prandtl Number}

An examination of the proof in the previous subsection reveals an interesting
fact that
the equation (\ref{ubar-eq}) for $\ol{\bu}$ involves both $\ol{\bu}$ and
$\ol{\bB},$ but
the equation (\ref{Bbar-eq}) for $\ol{\bB}$ involves only $\ol{\bB}$ itself.
This implies
that the stochastic representation previously employed for both $\bu$ and $\bB$
can be exploited for $\bB$ alone and, furthermore, at any magnetic Prandtl
number.
A precise statement of the result is as follows:

\begin{Prop}\lb{Prop2}
Divergence-free fields  $\bu,\bB\in C([t_0,t_f],C^{k,\alpha}(\Omega))$  satisfy
the
non-ideal, incompressible MHD equations (\ref{u-eq})-(\ref{B-eq})
with initial data $\bu_0,\bB_0\in C^{k,\alpha}(\Omega)$ for $k\geq 3$ iff the
momentum
equation (\ref{u-eq})
holds over that interval and simultaneously the stochastic flux conservation
holds
\be \int_S \bB(\bx,t)\bdot d\bS(\bx)=
     {\mathbb E} \left[\int_{\wt{\ba}(S,t)} \bB_0(\ba)\bdot d\bS(\ba) \right],
     \lb{Alfven-thm-Pr} \ee
for all smooth surfaces $S$ and all times $t\in [t_0,t_f],$,  where
$\wt{\ba}(\bx,t)$ are
``back-to-label maps'' for stochastic forward flows $\wt{\bx}(\ba,t)$ solving
the SDE
\be
d\wt{\bx}(\ba,t)=\bu(\wt{\bx}(\ba,t),t)dt+\sqrt{2\lambda}\,d\bW(t),\,\,t>t_0,
\,\,\,\,\,\,\,\,\,\,\,\,\,\,\, \wt{\bx}(\ba,t_0)=\ba. \lb{forward-Pr} \ee
\end{Prop}

\begin{proof}[Proof of Proposition \ref{Prop2}:]
The argument is nearly the same as that for the previous proposition. The
``if'' direction
is immediate, since the result (\ref{Alfven-thm-Pr}) is equivalent to the
stochastic
Lundquist formula (\ref{Lundquist}) and the equation (\ref{B-eq}) for $\bB$
follows
from (\ref{Lundquist}) using the generalized Ito rule, just as before. For the
``only if''
direction, we define
$$ \ol{\bB}(\bx,t)=
{\mathbb E} \left[\left.\bB_0(\ba)\bdot
\grad_a\wt{\bx}(\ba,t)\right|_{\wt{\ba}(\bx,t)} \right] $$
where $\wt{\bx}(\ba,t)$ is the stochastic flow defined by the SDE
(\ref{forward-Pr}) for the
velocity $\bu$ that satisfies the MHD momentum equation (\ref{u-eq}). It
follows that $\ol{\bB}$
satisfies the kinematic dynamo equation
$$  \partial_t \ol{\bB} =  -(\bu\bdot\grad)\ol{\bB}
        + \ol{\bB}\bdot\grad\bu+\lambda\bigtriangleup\ol{\bB},  \,\,t>t_0,
\,\,\,\,\,\,\,\,\,\,\,\,\,\,\, \ol{\bB}(t_0)=\bB_0 $$
one of whose solutions is $\ol{\bB}=\bB.$ Unicity of this solution again
follows from
an energy inequality
$$  e^{-2\gamma(t_f-t_0)}\|\ol{\bB}(t_f)\|_2^2+
       2\int_{t_0}^{t_f} dt\,  e^{-2\gamma(t-t_0)}
       \left[
\epsilon\|\ol{\bB}(t)\|_2^2+\lambda\|\grad\ol{\bB}(t)\|_2^2\right]
       \leq   \|\bB_0\|_2^2, $$
which is derived by a similar calculation as before \cite{Nunez97}, with
$\gamma>\sup_{t\in [t_0,t_f]}\|\grad\bu(t)\|_\infty + \epsilon.$
\end{proof}

\newpage

\section{Other Incompressible Plasma Fluid Models}\lb{Other}

We now establish similar stochastic conservation laws for some other non-ideal
plasma fluid
models, more refined than standard MHD. Keeping within the stated limitations
of this paper,
we consider here only the versions of these models assuming incompressible
fluid velocities
and isotropic transport coefficients. We also give only sketches of the proofs
of the stated
theorems, emphasizing essential differences from those given previously, since
most
of the details are very similar.


\subsection{Hall Magnetohydrodynamics}\lb{HMHD}

The equations of incompressible Hall magnetohydrodynamics (HMHD) have the form:
\be \partial_{t} \bu + (\bu\cdot\grad)\bu = -\grad p +
\frac{1}{4\pi\rho}(\grad\btimes\bB)\btimes\bB + \nu\nabla^{2}\bu
\lb{Hall-u-eq} \ee
\be \partial_{t} \bB =
\grad\times\left[\left(\bu-\frac{\alpha}{4\pi\rho}\grad\btimes\bB\right)
\times\bB\right] + \lambda\bigtriangleup\bB
\lb{Hall-B-eq} \ee
\be \grad\cdot\bB=\grad\cdot\bu=0
\lb{Hall-div-free} \ee
The magnetic induction equation (\ref{Hall-B-eq}) contains a ``Hall drift
term'' proportional
to $\alpha=mc/e,$ whose importance was first emphasized by Lighthill
\cite{Lighthill60}
and which was subsequently extensively investigated; see \cite{Witalis86} and
\cite{Kulsrud05,Bellan06}.
The limit $\alpha\rightarrow 0$ formally recovers standard MHD. Mathematical
properties
of HMHD solutions (existence, regularity, etc.) are studied in
\cite{Nunez04,Nunez05}.

Before stating our new theorems, we must review some standard facts about HMHD
eqs.(\ref{Hall-u-eq})-(\ref{Hall-div-free}), for which, for example, see
\cite{Ruban99}.
If one introduces a vector potential $\bA$ for the magnetic field in Coloumb
gauge,
$\grad\bdot\bA=0,$ then, along with a corresponding scalar potential $\Phi,$ it
satisfies
\be \partial_t \bA =
\left(\bu-\frac{\alpha}{4\pi\rho}\grad\btimes\bB\right)\btimes \bB-\grad\Phi
      + \lambda\bigtriangleup\bA. \lb{Hall-A-eq} \ee
HMHD is a Hamiltonian fluid model with two canonical momenta
\be \bp_i=\bu+\alpha^{-1}\bA,\,\,\,\,\,\,\,\, \bp_e=-\alpha^{-1}\bA,
\lb{Hall-mom} \ee
which are both divergence-free, $\grad\bdot\bp_\sigma=0,\,\,\,\,\sigma=i,e.$
When $\nu=\lambda,$
these satisfy the equations
\be \partial_t \bp_\sigma = \bu_\sigma\btimes(\grad\btimes\bp_\sigma) -\grad
\pi_\sigma
      +\nu \bigtriangleup\bp_\sigma, \,\,\,\,\sigma=i,e \lb{Hall-p-eqs} \ee
with $\pi_i=p+(1/2)|\bu|^2+\alpha^{-1}\Phi, \pi_e=-\alpha^{-1}\Phi,$ and with
$\bu_i$ the ion fluid velocity and  $\bu_e$ the electron fluid velocity, given
by
\be  \bu_i=\bu, \,\,\,\,\,\,\,\,\,
\bu_e=\bu-\frac{\alpha}{4\pi\rho}\grad\btimes\bB. \lb{Hall-vel} \ee
Note that, if $\nu\neq \lambda,$ then the equation for $\bp_i$ would contain an
additional term $(\nu-\lambda)\bigtriangleup\bp_e.$ Corresponding to the two
canonical momenta there are two  generalized vorticities
$\bOmega_\sigma=\grad\btimes\bp_\sigma,$
$\sigma=i,e,$ or concretely
\be \bOmega_i= \bomega + \alpha^{-1}\bB, \,\,\,\,\,\,\,\,
\bOmega_e=-\alpha^{-1}\bB.
\lb{Hall-vort} \ee
When $\nu=\lambda,$ these generalized vorticities satisfy
\be \partial_t \bOmega_\sigma = \grad\times
(\bu_\sigma\times\bOmega_\sigma)+\nu\bigtriangleup\bOmega_\sigma,
\,\,\,\,\sigma=i,e. \lb{Hall-Omega-eqs} \ee
These equations imply two ``frozen-in'' fields for ideal HMHD, one for the ion
fluid
and one for the electron fluid, and two Cauchy-type formulas for the two
generalized
vorticities. There are likewise two Kelvin-type theorems and two Weber formulas
for
the two canonical momenta.

We now state stochastic analogues of these results for non-ideal HMHD, first
for unit Prandtl number:

\begin{Prop321}\lb{Prop321}
Divergence-free fields  $\bu,\bB\in C([t_0,t_f],C^{k,\alpha}(\Omega))$  satisfy
the non-ideal,
incompressible HMHD equations (\ref{Hall-u-eq})-(\ref{Hall-B-eq}) with initial
data
$\bu_0,\bB_0\in C^{k,\alpha}(\Omega),$ with $k\geq 3$ and for unit magnetic
Prandtl
number $\nu/\lambda=1,$  iff for all closed, rectifiable loops $C$ and for all
$t\in [t_0,t_f]$
\be \oint_C \bp_\sigma(\bx,t)\bdot d\bx=
     {\mathbb E} \left[\oint_{\wt{\ba}_\sigma(C,t)} \bp_{\sigma\,0}(\ba)\bdot
d\ba \right],
     \,\,\,\,\,\,\,\,\sigma=i,e \lb{Hall-Kelvin-thms} \ee
Here the canonical momenta $\bp_\sigma,$ $\sigma=i,e$ are given by
eq.(\ref{Hall-mom}) and $\wt{\ba}_\sigma(\bx,t)$
are ``back-to-label maps'' for stochastic forward flows
$\wt{\bx}_\sigma(\ba,t)$  solving,
\be
d\wt{\bx}_\sigma(\ba,t)=
\bu_\sigma(\wt{\bx}_\sigma(\ba,t),t)dt+\sqrt{2\nu}\,d\bW(t),\,\,t>t_0,
\,\,\,\,\,\,\,\,\,\,\,\,\,\,\, \wt{\bx}_\sigma(\ba,t_0)=\ba,
\lb{Hall-stoch-flows} \ee
for $\sigma=i,e,$  with velocities $\bu_\sigma$ given in eq.(\ref{Hall-vel})
and $\bW(t)$
a standard Brownian motion.
\end{Prop321}

\noindent
{\bf Remarks:} {\it (i)} If the noise terms in eq.(\ref{Hall-stoch-flows}) were
chosen to be instead
$\sqrt{2\nu_\sigma}\,d\bW(t),\,\,\,\,\sigma=i,e$ with $\nu_i=\nu$ and
$\nu_e=\lambda,$ then one would
obtain the correct induction equation (\ref{Hall-B-eq}) but the momentum
equation would differ
from (\ref{Hall-u-eq}), containing an additional term
$\alpha^{-1}(\nu-\lambda)\bigtriangleup\bA
=\frac{4\pi e}{mc^2}(\lambda-\nu) \bJ.$

\noindent {\it (ii)} As we shall see below, it would be enough to assume
$\bu\in C([t_0,t_f],
C^{k,\alpha}(\Omega))$ for $k\geq 2,$ whereas $k\geq 3$ is required for $\bB$
because
of the Hall drift term.

\noindent {\it (iii)} The two stochastic Kelvin theorems in
(\ref{Hall-Kelvin-thms}) are equivalent
to {\it stochastic Weber formulas}
\be  \bp_\sigma(\bx,t) =
{\mathbb E }{\mathbb P}\left[ \grad_x\wt{\ba}_\sigma(\bx,t)
\,\bp_{\sigma\,0}(\wt{\ba}_\sigma(\bx,t))\right],\,\,\,\,\sigma=i,e.
\lb{Hall-Weber-frml} \ee
There are likewise two {\it stochastic Cauchy formulas} for the two generalized
vorticities:
\be \bOmega_\sigma(\bx,t)=
{\mathbb E} \left[\left.\bOmega_{\sigma\,0}(\ba)\bdot
\grad_a\wt{\bx}_\sigma(\ba,t)
           \right|_{\wt{\ba}_\sigma(\bx,t)} \right], \,\,\,\,\sigma=i,e.
\lb{Hall-Cauchy-frml} \ee

\begin{proof}[Sketch of proof of Proposition III.1.1]:  The proof is very
similar to those given in
the previous section and in \cite{ConstantinIyer08}. The main step is to derive
equations
for the stochastic time-differential of the variables
\be \wt{\bp}_\sigma={\mathbb P}\left[ \grad_x\wt{\ba}_\sigma
(\bp_{\sigma\,0}\circ\wt{\ba}_\sigma) \right]
                       = \grad_x\wt{\ba}_\sigma
(\bp_{\sigma\,0}\circ\wt{\ba}_\sigma) -\grad_x\wt{\varphi_\sigma}
\lb{ptilde-def} \ee
for $\sigma=i,e,$ of the form
\begin{eqnarray}
 d \wt{\bp}_\sigma(\bx,t) & = & \left[ -(\bu_\sigma\bdot\grad_x)\wt{\bp}_\sigma
-(\grad_x\bu_\sigma)\wt{\bp}_\sigma+\nu\bigtriangleup\wt{\bp}_\sigma\right]dt
              -\sqrt{2\nu}(d\bW(t)\bdot\grad_x)\wt{\bp}_\sigma \cr
       & & \,\,\,\,\,\,\,\,\,\,\,\,\,\,\, -\grad_x\left[d
\wt{\varphi}_\sigma+\left((\bu_\sigma\bdot\grad_x)\wt{\varphi}_\sigma
-\nu\bigtriangleup\wt{\varphi}_\sigma\right)dt
+\sqrt{2\nu}(d\bW(t)\bdot\grad_x)\wt{\varphi}_\sigma \right].
\lb{ptilde-eq} \end{eqnarray}
The calculations are essentially identical to those presented before. There is
just one technical
issue related to regularity of the field $\wt{\bp}_e,$ which should belong to
$C^{2,\alpha}(\Omega)$
in order to give classical meaning to the Laplacian term
$\nu\bigtriangleup\wt{\bp}_e.$ However,
if $\bu,\bB\in C^{3,\alpha}(\Omega)$, then $\bu_i\in
C^{3,\alpha}(\Omega),\bu_e\in C^{2,\alpha}
(\Omega),$ so that $\wt{\bx}_i,\wt{\ba}_i\in
C^{3,\alpha}(\Omega),\wt{\bx}_e,\wt{\ba}_e\in C^{2,\alpha}
(\Omega),$ and thus $\grad_x\wt{\ba}_i\in
C^{2,\alpha}(\Omega),\grad_x\wt{\ba}_e\in C^{1,\alpha}
(\Omega).$ This means that $\wt{\bp}_e$ defined by (\ref{ptilde-def}) belongs a
priori only to
$C^{1,\alpha}(\Omega).$ However, we may use an integration-by-parts identity
$$ {\mathbb P}\left[ (\grad_x \phi)\psi \right]= -{\mathbb P}\left[
\phi(\grad_x \psi)\right],
\,\,\,\,\,\,\,\,\,\,\,\phi,\psi\in C^{1,\alpha}(\Omega), $$
proved as Lemma 3.1 of \cite{Iyer06b}, in order to rewrite the partial
derivative $\partial_k\wt{\bp}_e$
in terms of $\grad_x\wt{\ba}_e$ only and eliminate second derivatives of
$\wt{\ba}_e.$ We can thus
conclude that  $\wt{\bp}_e\in C^{2,\alpha}(\Omega).$

Ensemble-averaging the equations (\ref{ptilde-eq}) for $\sigma=i,e$ yields
\be  \partial_t\ol{\bp}_\sigma =
-(\bu_\sigma\bdot\grad_x)\ol{\bp}_\sigma-(\grad_x\bu_\sigma)\ol{\bp}_\sigma
      -\grad_x\ol{\pi}_\sigma+\nu\bigtriangleup\ol{\bp}_\sigma, \lb{pbar-eq}
\ee
with
$\ol{\pi}_\sigma=\partial_t\ol{\varphi}_\sigma
+(\bu_\sigma\bdot\grad_x)\ol{\varphi}_\sigma
      -\nu\bigtriangleup\ol{\varphi}_\sigma,$ for $\sigma=i,e,$ or, in terms of
$\ol{\bu}$ and $\ol{\bB}$
variables,
\be \partial_t  \ol{\bu} =-(\bu\cdot\grad)\ol{\bu} -(\grad\bu)\ol{\bu}-\grad
\ol{p} +
     \frac{1}{4\pi\rho}(\grad\btimes\bB)\btimes\ol{\bB} +
\nu\bigtriangleup\ol{\bu},
     \lb{Hall-ubar-eq} \ee
\be \partial_t \ol{\bB} =
\grad\times\left[\left(\bu-\frac{\alpha}{4\pi\rho}\grad\btimes\bB\right)
\times\ol{\bB}\right] + \nu\bigtriangleup\ol{\bB}.  \lb{Hall-Bbar-eq} \ee
A fixed point satisfying $(\ol{\bu},\ol{\bB})=(\bu,\bB)$ must therefore also
obey the
HMHD equations (\ref{Hall-u-eq})-(\ref{Hall-B-eq}). The converse statement is
obtained
from the unicity of solutions to the initial-value problem for the above linear
equations
---either (\ref{pbar-eq}), $\sigma=i,e$ or
(\ref{Hall-ubar-eq}),(\ref{Hall-Bbar-eq})---which is
proved using energy estimates as before.
\end{proof}

\noindent Just as in the case of standard MHD, we see that the equation
(\ref{Hall-Bbar-eq})
for $\ol{\bB}$ does not depend upon $\ol{\bu}.$ This makes it possible to
derive a
stochastic conservation law for magnetic-flux at any Prandtl number:

\newpage

\begin{Prop322}\lb{Prop322}
Divergence-free fields  $\bu,\bB\in C([t_0,t_f],C^{k,\alpha}(\Omega))$  satisfy
the
non-ideal, incompressible HMHD equations (\ref{Hall-u-eq})-(\ref{Hall-B-eq})
with initial data $\bu_0,\bB_0\in C^{k,\alpha}(\Omega)$ for $k\geq 4$ iff the
momentum
equation (\ref{Hall-u-eq}) holds over that interval and simultaneously the
stochastic flux
conservation holds
\be \int_S \bB(\bx,t)\bdot d\bS(\bx)=
     {\mathbb E} \left[\int_{\wt{\ba}(S,t)} \bB_0(\ba)\bdot d\bS(\ba) \right],
     \lb{Hall-Alfven-thm} \ee
for all smooth surfaces $S$ and all times $t\in [t_0,t_f],$,  where
$\wt{\ba}(\bx,t)$ are
``back-to-label maps'' for stochastic forward flows $\wt{\bx}(\ba,t)$ solving
the SDE
\be
d\wt{\bx}(\ba,t)=\bu_e(\wt{\bx}(\ba,t),t)dt+\sqrt{2\lambda}\,d\bW(t),\,\,t>t_0,
\,\,\,\,\,\,\,\,\,\,\,\,\,\,\, \wt{\bx}(\ba,t_0)=\ba \lb{Hall-stoch-flow} \ee
with $\bu_e$ the electron fluid velocity given by (\ref{Hall-vel}).
\end{Prop322}

\noindent The proof is as in Proposition \ref{Prop2}, but greater smoothness of
$\bB$
is required to guarantee that $\wt{\bB}\in C^{2,\alpha}.$

\subsection{Two-Fluid Plasma Model}\lb{2F}

The most general hydrodynamical model of a fully ionized plasma consisting of
electrons and
one species of singly-charged ions is the two-fluid model of Braginsky
\cite{Braginsky65},
or the {\it Braginsky equations}. See also \cite{Kulsrud05,Bellan06}. The basic
variables of
the model are the two fluid velocities $\bu_\sigma,\,\,\,\,\sigma=i,e,$ with
$\sigma=i$ for the ion fluid
and $\sigma=e$ for the electron fluid. In the simple form considered here the
equations take the form:
\be (\partial_{t}+ \bu_i\bdot\grad)\bu_i =
+\frac{e}{m_i}\left(\bE+\frac{1}{c}\bu_i\btimes\bB\right)-\grad p_i
+ \nu_i\bigtriangleup \bu_i-\frac{1}{\tau_i}(\bu_i-\bu_e),
\,\,\,\,\,\,\,\,\,\,\,\grad\bdot\bu_i=0 \lb{2F-ui-eq} \ee
\be (\partial_{t}+ \bu_e\bdot\grad)\bu_e =
-\frac{e}{m_e}\left(\bE+\frac{1}{c}\bu_e\btimes\bB\right)-\grad p_e
+ \nu_e\bigtriangleup \bu_e-\frac{1}{\tau_e}(\bu_e-\bu_i),
\,\,\,\,\,\,\,\,\,\,\,\grad\bdot\bu_e=0 \lb{2F-ue-eq} \ee
\be -\bigtriangleup\bA = \frac{4\pi}{c}\bJ =
\frac{4\pi}{c}ne(\bu_i-\bu_e),\,\,\,\,\,\,\,\,\,\,\,\grad\bdot\bA=0
\lb{2F-A-eq} \ee
\be \bE=-\frac{1}{c}\partial_t\bA, \,\,\,\,\,\,\,\,\,\,\,\bB=\grad\btimes\bA
\lb{homo-max} \ee
In addition to internal viscosities $\nu_\sigma,\,\,\sigma=i,e,$ there are
linear drag terms which represent the
exchange of momentum between the two fluids by collisions of the constituent
particles and which
are proportional to the collision frequencies $1/\tau_\sigma,\,\,\sigma=i,e.$
Conservation of momentum requires
$m_e/\tau_e=m_i/\tau_i.$ Note that the vector potential $\bA$ is not an
independent variable, but is
completely determined from $\bu_e,\bu_i$ by means of the elliptic equation
(\ref{2F-A-eq}).
Mathematical properties of solutions of these equations (existence, regularity,
etc.) are
studied in \cite{Nunez08}, including even a separate equation for neutral
molecules.

Neglecting viscosities and drag, the two-fluid model is Hamiltonian with
canonical
momenta:
\be \bp_i = \bu_i + \frac{e}{m_i c}\bA,\,\,\,\,\,\,\,\,\,\,\,\bp_e = \bu_e -
\frac{e}{m_e c}\bA,
\lb{2F-p-def} \ee
satisfying $\grad\bdot\bp_\sigma=0$ for $\sigma=i,e.$ E.g. see \cite{Ruban99}.
Using
$(\bu_\sigma \bdot\grad)\bu_\sigma =\grad(\frac{1}{2}|\bu_\sigma
|^2)-\bu_\sigma \btimes
(\grad\btimes\bu_\sigma ),$ $\pi_\sigma =p_\sigma  +\frac{1}{2}|\bu_\sigma
|^2,$ the
eqs.(\ref{2F-ui-eq}),(\ref{2F-ue-eq}) for $\bu_e,\bu_i$ can be rewritten for
$\bp_e,\bp_i$ as
\be \partial_t\bp_i = \bu_i\btimes(\grad\btimes\bp_i)-\grad \pi_i
+ \nu_i\bigtriangleup \bu_i-\frac{1}{\tau_i}(\bu_i-\bu_e). \lb{2F-pi-eq} \ee
\be \partial_t\bp_e = \bu_e\btimes(\grad\btimes\bp_e)-\grad \pi_e
+ \nu_e\bigtriangleup \bu_e-\frac{1}{\tau_e}(\bu_e-\bu_i), \lb{2F-pe-eq} \ee
Define the magnetic diffusivity
 $\lambda=m_i c^2/4\pi ne^2\tau_i=m_e c^2/4\pi ne^2\tau_e,$
so that
$ \frac{m_i}{\tau_i}(\bu_e-\bu_i)=\frac{m_e}{\tau_e}(\bu_e-\bu_i)=
     \lambda\cdot \frac{e}{c}\bigtriangleup\bA. $
Hence, choosing $\nu_e=\nu_i=\lambda,$ (\ref{2F-pi-eq}),(\ref{2F-pe-eq}) become
\be \partial_t\bp_\sigma  = \bu_\sigma \btimes(\grad\btimes\bp_\sigma )-\grad
\pi_\sigma
+ \lambda\bigtriangleup \bp_\sigma , \,\,\,\,\,\,\,\,\,\,\,\sigma=i,e.
\lb{2F-p-eqs} \ee
The vector potential $\bA$ can be recovered from $\bp_i,\bp_e$ by solving the
Helmholtz equation
\be -\bigtriangleup\bA + \kappa^2\bA=
\frac{4\pi}{c}ne(\bp_i-\bp_e),\,\,\,\,\,\,\,\,\,\,\,\grad\bdot\bA=0
\lb{2F-A-eq-pversion} \ee
with $\kappa^2=4\pi n e^2/\mu c^2$ and $\mu^{-1}=m_i^{-1}+m_e^{-1}$ and then
$\bu_i,\bu_e$ obtained from (\ref{2F-p-def}).  For this non-ideal version of
the two-fluid
model there are two stochastic conservation laws corresponding to the two
canonical momenta:

\begin{Prop33}\lb{Prop33}
Divergence-free fields  $\bu_e,\bu_i\in C([t_0,t_f],C^{k,\alpha}(\Omega))$
satisfy the
non-ideal, incompressible two-fluid equations (\ref{2F-ui-eq})-(\ref{homo-max})
with initial data
$\bu_{e\,0},\bu_{i\,0}\in C^{k,\alpha}(\Omega)$ for $k\geq 2$ and for unit
magnetic Prandtl
numbers $\nu_e/\lambda=\nu_i/\lambda=1,$  iff for all closed, rectifiable loops
$C$ and for all
$t\in [t_0,t_f]$
\be \oint_C \bp_\sigma (\bx,t)\bdot d\bx=
     {\mathbb E} \left[\oint_{\wt{\ba}_\sigma (C,t)} \bp_{\sigma\,0}(\ba)\bdot
d\ba \right],
     \,\,\,\,\,\,\,\,\sigma=i,e \lb{2F-Kelvin-thms} \ee
Here the canonical momenta $\bp_\sigma ,$ $\sigma=i,e$ are given by
eq.(\ref{2F-p-def})
and $\wt{\ba}_\sigma (\bx,t)$ are ``back-to-label maps'' for stochastic forward
flows
$\wt{\bx}_\sigma (\ba,t)$  solving,
\be d\wt{\bx}_\sigma (\ba,t)=\bu_\sigma (\wt{\bx}_\sigma
(\ba,t),t)dt+\sqrt{2\lambda}\,d\bW(t),\,\,t>t_0,
\,\,\,\,\,\,\,\,\,\,\,\,\,\,\, \wt{\bx}_\sigma (\ba,t_0)=\ba,
\lb{2F-stoch-flows} \ee
for $\sigma=i,e,$  with $\bW(t)$ a standard Brownian motion.
\end{Prop33}

\noindent
{\bf Remarks:} {\it (i)} If the noise terms in eq.(\ref{2F-stoch-flows}) were
chosen to be
$\sqrt{2\lambda_\sigma }\,d\bW_\sigma (t),\,\,\,\,\sigma=i,e$ with
$\lambda_e\neq\lambda_i$ then one would
obtain two-fluid model (\ref{2F-ui-eq})-(\ref{homo-max}) with $\nu_\sigma
=\lambda_\sigma $ and
$\tau_\sigma =m_\sigma c^2/4\pi ne^2\lambda_\sigma $  for $\sigma=i,e.$
Although mathematically well-posed,
this system is unphysical since it violates conservation of momentum.

\noindent {\it (ii)} The two stochastic Kelvin theorems (\ref{2F-Kelvin-thms})
are mathematically
equivalent to {\it stochastic Weber formulas}:
\be  \bp_\sigma (\bx,t) =
{\mathbb E }{\mathbb P}\left[ \grad_x\wt{\ba}_\sigma (\bx,t)
\,\bp_{\sigma\,0}(\wt{\ba}_\sigma (\bx,t))\right],\,\,\,\,\sigma=i,e,
\lb{2F-Weber-frml} \ee
identical in form to (\ref{Hall-Weber-frml}) for HMHD.

\begin{proof}[Sketch of proof of Proposition \ref{Prop33}] The proof is very
similar to
that of Proposition III.1.1 and to the proofs in Section 4 of
\cite{ConstantinIyer08}.
Stochastic Weber variables $\wt{\bp}_\sigma ,\,\,\sigma=i,e$ of the same form
as (\ref{ptilde-def}) are
shown to obey stochastic PDE's of the same form as (\ref{ptilde-eq}). It is now
enough
to assume $\bu_e,\bu_i\in C([t_0,t_f],C^{k,\alpha}(\Omega))$ for $k\geq 2,$
because
the integration-parts-identity \cite{Iyer06a,Iyer06b} can be employed to show
that
$\wt{\bp}_e,\wt{\bp}_i\in C^{2,\alpha}(\Omega)),$ just as for $\wt{\bp}_i$ in
the proof
of Proposition III.1.1.

\end{proof}

\noindent The curl of the two canonical momenta in (\ref{2F-p-def}) give two
generalized
vorticities, $\bOmega_\sigma =\grad\btimes\bp_\sigma$ for $\sigma=i,e$:
\be \bOmega_i = \bomega_i  + \frac{e}{m_i c}\bB, \,\,\,\,\,\,\,\,\,\,\,
\bOmega_e =
      \bomega_e  - \frac{e}{m_e c}\bB.   \lb{2F-gen-vort} \ee
If we assume sufficient smoothness ($\bu_e,\bu_i\in
C([t_0,t_f],C^{k,\alpha}(\Omega))$ with
$k\geq 3$), then these satisfy equations
\be  \partial_t \bOmega_\sigma
= \grad\times (\bu_\sigma \times\bOmega_\sigma
)+\lambda\bigtriangleup\bOmega_\sigma ,
\,\,\,\,\sigma=i,e, \lb{2F-Helm-eq} \ee
and there are two stochastic Cauchy formulas
\be  \bOmega_\sigma (\bx,t)=
{\mathbb E} \left[\left.\bOmega_{\sigma\,0}(\ba)\bdot \grad_a\wt{\bx}_\sigma
(\ba,t)
           \right|_{\wt{\ba}_\sigma (\bx,t)} \right], \,\,\,\,\sigma=i,e,
\lb{2F-Cauchy-frml} \ee
of the same form as (\ref{Hall-Cauchy-frml}) for HMHD. Unlike for the previous
models, for the
two-fluid model there is no separate stochastic ``frozen-in'' property for the
magnetic field
$\bB$ alone, at general Prandtl number.

\section{Discussion}\lb{Discuss}

The results of the present paper provide new tools with which to investigate
and explain
resistive phenomena in plasma fluids. The stochastic Lagrangian conservation
laws
derived here are the analogues for non-ideal hydromagnetic systems of
flux-freezing
for ideal ones. Especially robust is the stochastic Alfv\'{e}n theorem and
stochastic
Lundquist formula, which hold in both MHD and HMHD at any magnetic Prandtl
number.
These results have important implications for resistive magnetic reconnection
and
related problems such as magnetic dynamo \cite{Kulsrud05,Bellan06}, which will
be
pursued in detail in future publications. Here we just note that that the
stochastic
Lundquist formula  describes how the resultant magnetic field $\bB(\bx,t)$ at a
spacetime point $(\bx,t)$ is obtained by advecting all magnetic field lines as
``frozen-in'' to the stochastic flows, with added white-noise, and then
averaging
those magnetic field vectors that arrive to the given point. The advection by
the stochastic
flows produces the nonlinear effects of magnetic stretching and tilting, while
the average
over the Brownian motions represents the resistive ``gluing'' of the magnetic
field,
reconnecting the field-lines and changing their topology. Such resistive
effects are
recognized as important  in the dynamo process by the cycle of
``stretch-twist-fold-reconnect''
\cite{Galloway03}. Moffatt has referred to the ``oxymoronic role'' of
resistivity,
writing that ``the dynamo process may be described as a process  of
`regenerative decay',
or perhaps better `reinvigorating dissipation'.'' \cite{Moffatt78}

A possible criticism of the physical relevance of our results is that molecular
resistivity,
represented by a Laplacian term in the induction equation,  is a poor model
of actual dissipative processes in a plasma. In contrast to the iconic
status of the viscosity term in the Navier-Stokes equation for neutral fluids,
there
is considerably less universality in the form of the dissipation in plasmas or,
indeed, in the validity of a hydromagnetic description. The two-fluid equations
of Braginsky \cite{Braginsky65} (see also \cite{Kulsrud05,Bellan06})are more
complicated than those discussed in our section \ref{2F}. For example,
viscosity and
resistivity in the standard Braginsky equations are anisotropic, with
magnitudes
differing along directions longitudinal and transverse to the local magnetic
field.
Furthermore, microscopic Spitzer resistivity is not the only form of magnetic
dissipation
that may occur in plasmas. A wide variety of processses, both collisional and
non-collisional, have been proposed to lead to ``anomalous resistivity'' of
different
forms \cite{Papadopoulos77,Treumann01}. Furthermore, in a partially ionized
plasma
the collisions of ions with neutral molecules induces an ``ambipolar drift''
of magnetic field lines with velocity proportional to the Lorenz force
\cite{Spitzer78,BrandenburgSubramanian05} and this can be the most significant
form
of magnetic dissipation in some cases, e.g. the interstellar medium.  Thus, the
fluid models that we have considered are not necessarily the most physically
realistic.

There are two responses that we can give to this important set of criticisms.

First, the results presented in this paper are far from the most general
possible.
We have chosen to restrict discussion here to models with incompressible fluids
and isotropic transport coefficients, since these hydromagnetic models are
widely employed and the proofs of the main results are simpler for them than
for more complete models.  However, in a following work \cite{EyinkNeto09}, we
establish similar results for much more general plasma fluid models, allowing
for compressible fluids,  anisotropic pressure and transport coefficients,
neutral components, etc. Stochastic conservation laws of the sort demonstrated
here are quite general and should hold for a very large class of non-ideal
plasma fluid models, when the ideal version of the model possesses a
corresponding
``frozen-in'' field. It may even be possible to prove similar stochastic laws
for
kinetic models of plasmas with collisions described by Boltzmann kernels
\cite{FournierMeleard01} or Fokker-Planck operators \cite{Kulsrud05} (Section
8.3),
since the ideal, collisionless Vlasov dynamics possesses analogues of the
frozen-in invariants \cite{Yankov97}.

Second, the precise form of the dissipation in hydromagnetic models may not
matter, as long as its effects are confined to sufficiently small
length-scales.
There is then a large ``effective Reynolds number'' (both magnetic and kinetic)
and the plasma fluid becomes turbulent. We have previously argued
\cite{Eyink07} that the laws of flux conservation and magnetic line-motion
in hydromagnetic turbulence are intrinsically stochastic in the limit of
infinite
Reynolds number. Formally, the random white-noise disappears in the
equations for stochastic Lagrangian particles
\be d\wt{\bx} = \bu^\nu(\wt{\bx}(t),t)dt+\sqrt{2\lambda}\, d\bW(t), \,\,t>t_0,
\,\,\,\,\,\,\,\,\,\,\,\,\,\,\, \wt{\bx}(t_0)=\bx_0 \lb{particle-eq} \ee
as $\nu,\lambda\rightarrow 0$ (cf. also eq.(\ref{forward})). However, the
randomness
need not vanish if the advecting velocity $\bu^\nu$ solving (\ref{u-eq})
approaches a rough
or singular velocity $\bu$ in this limit,  as expected for a Kolmogorov-type
cascade range.
As a consequence of ``explosive'' separation of particles in Richardson
two-particle turbulent
diffusion, a pair of solutions of (\ref{particle-eq}) with the {\it same}
initial condition
$\bx_0$ may separate at time $t$ to a mean-square distance $\sim t^3$ in the
limit
$\nu,\lambda\rightarrow 0.$ These statements have been proved rigorously to
hold
in the Kazantsev-Kraichnan kinematic dynamo model
\cite{Kazantsev68,Kraichnan68}.
It has furthermore been proved that the Lagrangian trajectories in the
Kazantsev-Kraichnan
model remain stochastic as $\nu,\lambda\rightarrow 0,$ a result that has been
termed ``spontaneous stochasticity.'' The limiting probability distributions of
trajectories
are known to be very robust and universal for the case of an incompressible
fluid velocity,
with the same result being obtained for limits of a wide class of
regularizations.  See
\cite{Eyink06,Eyink07} for references and more detailed discussion. The
rigorous
results for the Kazantsev-Kraichnan dynamo model and the new results in the
present
work give further plausibility to the ideas that the precise form of
dissipation
does not matter in nonlinear hydromagnetic turbulence and that flux
conservation
and ``frozen-in'' line-motion will remain as stochastic laws in the limit of
very large
Reynolds numbers.

\vspace{.25in}

 {\small
\noindent {\bf Acknowledgements}. We acknowledge the warm hospitality of the
Isaac
Newton Institute for Mathematical Sciences during the programme on ``The Nature
of
High Reynolds Number Turbulence'', when this paper was completed. This work was
partially supported by NSF grant AST-0428325 at Johns Hopkins University. }


\newpage

\begin{thebibliography}{99}

\bibitem{Alfven42}
Alfv\'{e}n, H., ``On the existence of electromagnetic-hydrodynamic waves,''
Arkiv
f. Mat., Astron. o. Fys. {\bf 29B} 1--7 (1942).
\bibitem{BekensteinOron00}
Bekenstein, J. D.  and A. Oron,  ``Conservation of circulation in
magnetohydrodynamics,''  Phys. Rev. E {\bf 62} 5594--5602 (2000).
\bibitem{Bellan06}
Bellan, P. M., {\it Fundamentals of Plasma Physics.} (Cambridge University
Press,
Cambridge, UK, 2006)
\bibitem{Braginsky65}
Braginsky,  S. I.,  ``Transport processes in a plasma,'' Rev. Plasma Phys. {\bf
1}
205--311 (1965).
\bibitem{BrandenburgSubramanian05}
Brandenburg, A.  and K. Subramanian, ``Astrophysical magnetic fields and
nonlinear
dynamo theory,'' Phys. Rep. {\bf 417} 1--209 (2005).
\bibitem{Cauchy1815}
Cauchy,  A. L.,``Th\'{e}orie de la propagation des ondes \`{a} la surface d'un
fluide pesant d'une
profondeur ind\'{e}finie (1815)'', M\'{e}m. Divers Savants (2) {\bf 1} 3;
Oeuvres (1) {\bf 1} 5.
\bibitem{ChorinMarsden00}
Chorin, A.  and J. Marsden, {\it Mathematical Introduction to Fluid Mechanics.}
(Berlin-Heidelberg-New York, Springer, 2000).
\bibitem{Constantin08}
Constantin,  P., ``Singular, weak and absent: Solutions of the Euler
equations,''
Physica D {\bf 237} 1926--1931 (2008)
\bibitem{ConstantinIyer08}
Constantin, P.  and G. Iyer,  ``A stochastic Lagrangian representation of the
three-dimensional
incompressible Navier-Stokes equations,'' Commun. Pure Appl. Math. {\bf LXI}
0330--0345 (2008).
\bibitem{DuvautLions72}
Duvaut, G.  and J. L. Lions, ``In\'{e}quations en thermo\'{e}lasticit\'{e} et
magn\'{e}tohydrodynamique,'' Arch. Ration. Mech. Anal. {\bf 46} 241--279
(1972).
\bibitem{Eyink06}
Eyink,  G. L., ``Turbulent cascade of circulations,'' C. R. Physique {\bf 7}
449--455 (2006).
\bibitem{Eyink07}
Eyink,  G. L., ``Turbulent diffusion of lines and circulations,'' Phys. Lett.
A. {\bf 368}
486--490 (2007).
\bibitem{Eyink08}
Eyink,   G. L., ``Dissipative anomalies in singular Euler flows,'' Physica D
{\bf 237} 1956--1968 (2008)
\bibitem{Eyink09}
Eyink, G. L., ``Stochastic least-action principle for the incompressible
Navier-Stokes
equation,'' Physica D, accepted (2008)
\bibitem{EyinkNeto09}
Eyink, G. L.  and A. F. Neto, ``Stochastic line-motion and stochastic
conservation laws
for non-ideal hydromagnetic models.  II. Compressible fluids and
anisotropic transport coefficients,'' in preparation (2009).
\bibitem{FournierMeleard01}
Fournier, N.  and S. M\'{e}l\'{e}ard, ``A Markov process associated with a
Boltzmann
equation without cutoff and for non-Maxwell molecules, J. Stat. Phys. {\bf 104}
359--385 (2001).
\bibitem{Galloway03}
Galloway, D., ``Fast dynamos,'' in: {\it Advances in Nonlinear Dynamos.}
eds. A. Ferriz-Mas and M. N\'{u}\~{n}ez (CRC Press, 2003).
\bibitem{Helmholtz1858}
Helmholtz, H., `` \"{U}ber Integrale der hydrodynamischen Gleichungen welche
den
Wirbelbewegungen entsprechen,'' Crelles Journal {\bf 55} 25--55 (1858).
\bibitem{Iyer06a}
Iyer,  G., ``A stochastic Lagrangian formulation of the Navier-Stokes and
related transport
equations.'' Doctoral dissertation, University of Chicago, 2006.
\bibitem{Iyer06b}
Iyer,  G., ``  A stochastic perturbation of inviscid flows,'' Comm. Math. Phys.
{\bf 266}
631--645 (2006)
\bibitem{Kazantsev68}
Kazantsev, A. P., ``Enhancement of a magnetic field by a conducting fluid,"
Sov. Phys.
JETP {\bf 26} 1031--1034 (1968).
\bibitem{Kraichnan68}
Kraichnan,  R. H., ``Small-scale structure of a scalar field convected by
turbulence,''
Phys. Fluids {\bf 11} 945--953 (1968).
\bibitem{Kulsrud05}
Kulsrud, R. M., {\it Plasma Physics for Astrophysics.} (Princeton University
Press,
Princeton, NJ, 2005)
\bibitem{Kunita90}
Kunita, H.,  {\it Stochastic Flows and Stochastic Differential Equations.}
(Cambridge University Press, Cambridge, 1990).
\bibitem{KuznetsovRuban00}
Kuznetsov, E. A.  and V. P. Ruban, ``Hamiltonian dynamics of vortex and
magnetic lines in hydrodynamic type systems'', Phys. Rev. E. {\bf 61}
831--841 (2000).
\bibitem{Lighthill60}
Lighthill, M. J., ``Studies on MHD waves and other anisotropic wave motion,''
Phil. Trans. Roy. Soc. {\bf 252A} 397--430 (1960).
\bibitem{Lundquist51}
Lundquist, S., ``On the stability of magneto-hydrostatic fields'', Phys. Rev.
{\bf 83}
307--311 (1951).
\bibitem{Moffatt78}
Moffatt,  H. K.,``The oxymoronic role of molecular diffusivity in the dynamo
process,''  Woods Hole Oceanographic Institution Technical Rep. WHOI-78-67,
145--149 (1978).
\bibitem{Newcomb58}
Newcomb, W. A., ``Motion of magnetic lines of force,'' Ann. Phys. N.Y. {\bf 3}
347--385 (1958).
\bibitem{Nunez97}
N\'{u}\~{n}ez,  M., ``Some rigorous results for the kinematic dynamo problem
with general boundary conditions,''  J. Math. Phys. {\bf 38} 1583--1592 (1997).
\bibitem{Nunez04}
N\'{u}\~{n}ez,  M., ``Growth of the magnetic field in Hall
magnetohydrodynamics,''
J. Phys. A: Math. Gen. {\bf 37} 9317--9323 (2004).
\bibitem{Nunez05}
N\'{u}\~{n}ez,  M., ``Existence theorems for two-fluid magnetohydrodynamics,''
J. Math. Phys. {\bf 46} 083101 (2005).
\bibitem{Nunez08}
N\'{u}\~{n}ez,  M., ``A theorem of existence for the equations of
magnetohydrodynamics
of partially ionized plasmas,''  Proc. R. Soc. A {\bf 464} 1571--1586 (2008).
\bibitem{Papadopoulos77}
Papadopoulos,  K., ``A review of anomalous resistivity for the ionosphere,''
Rev. Geophys. Sp. Phys. {\bf 15} 113--127 (1977).
\bibitem{Ruban99}
Ruban, V. P., ``Motion of magnetic flux lines in magnetohydrodynamics,'' JETP
{\bf 89}
299--310 (1999).
\bibitem{SermangeTemam83}
Sermange, M. and R. Temam, ``Some mathematical questions related to the MHD
equations,'' Commun. Pure Appl. Math. {\bf 36} 635--664 (1983).
\bibitem{Spitzer78}
Spitzer, Jr., L., {\it Physical Processes in the Interstellar Medium.} (J.
Wiley \& Sons, New
York, 1978)
\bibitem{Stern66}
Stern,  D. P., ``The motion of magnetic field lines,'' Space Science Reviews
{\bf 6}
147--173 (1966).
\bibitem{Thomson1869}
Thomson, W.  (Lord Kelvin), ``On vortex motion'', Trans. Roy. Soc. Edin.
{\bf 25} 217--260 (1869).
\bibitem{Treumann01}
Treumann,  R. A., ``Origin of resistivity in reconnection,'' Earth Planets
Space
(Japan) {\bf 53} 453--462 (2001).
\bibitem{Weber1868}
Weber,  W., ``\"{U}ber eine Transformation der hydrodynamischen Gleichungen,''
J. Reine Angew.Math.  {\bf 68} 286--292 (1868).
\bibitem{Witalis86}
Witalis, E. A., ``Hall magnetohydrodynamics and its applications to laboratory
and cosmic plasma,'' IEEE Trans. Plasma Science {\bf PS-14} 842--848 (1986).
\bibitem{Yankov97}
Yan'kov,  V. V., ``Attractors and frozen-in invariants in turbulent plasma,''
Phys. Uspekhi {\bf 40} 477--493 (1997).
 \end{thebibliography}
\end{document}